\documentclass[letterpaper, 10 pt, conference]{ieeeconf}

\IEEEoverridecommandlockouts        
\usepackage{cite}
\usepackage{amsmath,amssymb,amsfonts}
\usepackage{algorithmic}
\usepackage{hyperref}
\usepackage{epstopdf}
\usepackage{textcomp}
\usepackage{graphicx,color}
\usepackage{mathrsfs}
\usepackage[vlined,ruled]{algorithm2e}
\usepackage{subfigure}
\usepackage{url}
\usepackage{color}
\usepackage{dsfont}
\usepackage{bbm}
\usepackage{booktabs}
\usepackage{array}
\usepackage[table]{xcolor}
\usepackage{yfonts}
\usepackage{cleveref} 

\newtheorem{theorem}{Theorem}[section]

\newtheorem{remark}{Remark}

\newtheorem{appxlem}{Lemma}[section]
\newtheorem{appxthm}{Theorem}[section]

\newcommand{\setdef}[2]{\{#1 \; : \; #2\}}

\newcommand{\until}[1]{\{1,\dots,#1\}}

\newcommand{\real}{\mathbb{R}}

\newcommand{\transpose}{\mathsf{T}} 

\newcommand{\mc}{\mathcal}

\newcommand{\1}{\mathds{1} }

\DeclareSymbolFont{bbold}{U}{bbold}{m}{n}
\DeclareSymbolFontAlphabet{\mathbbold}{bbold}

\newcommand\oprocendsymbol{\hbox{$\square$}}
\newcommand\oprocend{\relax\ifmmode\else\unskip\hfill\fi\oprocendsymbol}

\renewcommand{\theenumi}{(\roman{enumi})}

\newcommand*{\QEDA}{\hfill\ensuremath{\blacksquare}}%

\DeclareMathOperator{\im}{i}

\renewcommand{\baselinestretch}{0.977}

\begin{document}
\title{\bf A Framework to Control Functional Connectivity in the Human
  Brain} \author{Tommaso Menara, Giacomo Baggio, Danielle S. Bassett,
  and Fabio Pasqualetti \thanks{This material is based upon work
    supported in part by ARO 71603NSYIP, and in part by NSF
    BCS1631112. Tommaso Menara, Giacomo Baggio and Fabio Pasqualetti
    are with the Department of Mechanical Engineering, University of
    California at Riverside,
    \{\href{mailto:tomenara@engr.ucr.edu}{\texttt{tomenara}},
    \href{mailto:gbaggio.ucr.edu}{\texttt{gbaggio}},
    \href{mailto:fabiopas@engr.ucr.edu}{\texttt{fabiopas\}@engr.ucr.edu.}}
    Danielle S. Bassett is with the Department of Bioengineering, the
    Department of Electrical and Systems Engineering, the Department
    of Physics and Astronomy, the Department of Psychiatry, and the
    Department of Neurology, University of Pennsylvania,
    \href{mailto:mailto:dsb@seas.upenn.edu}{\texttt{dsb@seas.upenn.edu.}}}}
\maketitle


\begin{abstract}
  In this paper, we propose a framework to control brain-wide
  functional connectivity by selectively acting on the brain's
  structure and parameters. Functional connectivity, which measures
  the degree of correlation between neural activities in different
  brain regions, can be used to distinguish between healthy and
  certain diseased brain dynamics and, possibly, as a control
  parameter to restore healthy functions. In this work, we use a
  collection of interconnected Kuramoto oscillators to model
  oscillatory neural activity, and show that functional connectivity
  is essentially regulated by the degree of synchronization between
  different clusters of oscillators. Then, we propose a minimally
  invasive method to correct the oscillators' interconnections and
  frequencies to enforce arbitrary and stable synchronization patterns
  among the oscillators and, consequently, a desired pattern of
  functional connectivity. Additionally, we show that our
  synchronization-based framework is robust to parameter mismatches
  and numerical inaccuracies, and validate it using a realistic
  neurovascular model to simulate neural activity and functional
  connectivity in the human brain.
\end{abstract}

\section{Introduction}\label{sec: introduction}
The structural (i.e., matrix of anatomical connections between brain
regions) and functional (i.e., matrix of correlation coefficients
between the activity of brain regions) connectivity of the brain vary
across healthy individuals and those affected by neurological or
psychiatric disorders, and can be used as biomarkers to detect or
predict pathological conditions. While structural connectivity changes
rather slowly over time and can be measured accurately via diffusion
imaging techniques \cite{HP-CL-GX-MR-HCJ-WVJ-SO:08}, functional
connectivity depends on the instantaneous neural activity and is
affected, for instance, by the tasks being performed and external
stimuli \cite{AZ-AF-EB:12}. Today, common measures of functional
connectivity rely on resting-state functional magnetic resonance
imaging (rs-fMRI) timeseries to quantify the level of correlated
activity between brain regions. The relationships between structural
and functional connectivity have recently received considerable
attention \cite{CJH-RK-MB-OS:07, COB-SP-GJP-MBM-STG-DSB-VMP:18}, and
the tantalizing idea of controlling functional states by leveraging or
modifying brain structure has given birth to a new, thrilling, field
of research~\cite{SG-FP-MC-QKT-BYA-AEK-JDM-JMV-MBM-STG-DSB:15,
  TM-DSB-FP:17, JS-ANK-TM-AEK-JMS-SRD-RG-JT-BL-KAD-FP-TL-DSB:19}.

In this paper, we leverage the connection between structural and
functional connectivity, and propose a framework to control functional
connectivity by selectively modifying structural connectivity and the
regions' intrinsic frequencies (see Fig.~\ref{fig: FC control}). In
particular, building on prior studies
\cite{GD-VK-ARM-OS-RK:09,JC-EH-OS-GD:11}, we model the brain's neural
activity as the phases of a collection of interconnected Kuramoto
oscillators, and postulate that the level of functional connection
between two regions is proportional to the level of synchronization
between the phases of the oscillators associated with the two
regions. Then, we derive conditions and methods to tune the
oscillators' interconnection weights and natural frequencies so as to
enforce arbitrary synchronization patterns and, consequently,
brain-wide functional connectivity. We remark that the control
mechanisms used in our framework are biologically plausible.  For
instance, changes in the spontaneous neural activity (i.e.,
oscillators' frequencies) are typical of the brain, involve natural
modifications in regional metabolism of the neurons, and can
alternatively be induced by a number of non-invasive stimulation
techniques \cite{RP-MAN-CCR:18}. Changes to the
structural interconnections (i.e., oscillators' interconnections),
instead, can arise from different chemical or electrical mechanisms
including, at the microscale, Hebbian plasticity \cite{AK-MFB:94} and
short-term synaptic facilitation \cite{AMT:00}.

\begin{figure}[t]
  \centering
    \includegraphics[width=1\columnwidth]{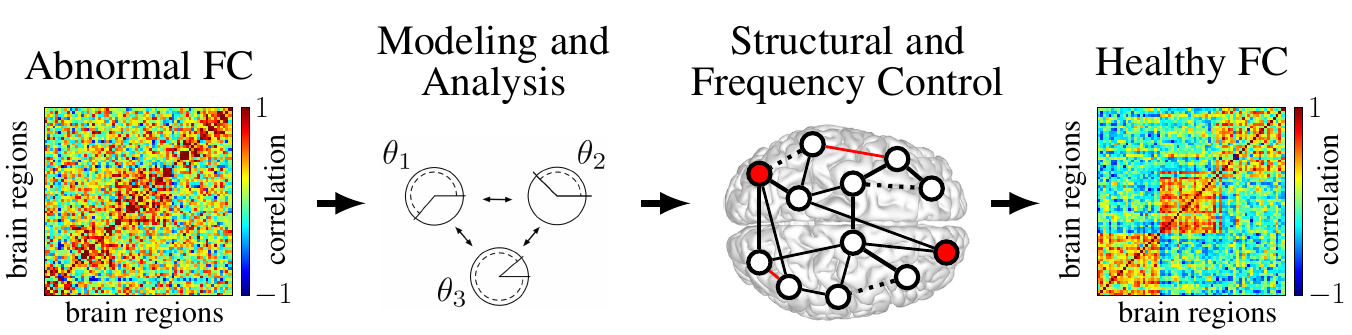}
    \caption{This paper proposes a framework to restore healthy
      patterns of brain-wide functional connectivity by selectively
      acting on the brain's structure and parameters. Using a network
      of heterogeneous Kuramoto oscillators to model the brain's
      neural activity, we design and validate a minimally invasive
      method to correct the oscillators' interconnections and
      frequencies to obtain a desired and stable pattern of functional
      connectivity.}
  \label{fig: FC control}
\end{figure}

\noindent\textbf{Related work.} The discovery of oscillatory
or rhythmic brain activity dates back almost a century. Yet,
control-theoretic studies that exhaust the oscillatory nature of brain
states have been sparse and of relatively recent date. Some authors
focus on localized desynchronization of neural activity
\cite{PAT:03,AF-AC-EP-FLL:12,Pecora2014}, which is desirable in
individuals affected by epilepsy or Parkinson's disease, and others
use synchronization phenomena to describe cognitive and functional
brain states \cite{MJ-XY-MP-HS-KHJ:18,YQ-YK-MC:18,EN-JC:19}. To the
best of our knowledge, a framework to control the pattern of
brain-wide functional connectivity is still missing, and is proposed
for the first time in this paper.

At the core of our framework to model and control functional
connectivity is the concept of \emph{cluster synchronization} in a
network of oscillators, where groups of oscillators behave cohesively
but independently from other clusters.  For the case of oscillators
with Kuramoto dynamics as used in this work,
\cite{CF-AC-FP:17,YQ-YK-MC:18b} explore approximate notions of cluster
synchronization in simplified configurations, while
\cite{LT-CF-MI-DSB-FP:17} provides exact invariance conditions for
arbitrary cluster synchronization manifolds. Our recent work
introduces rigorous \cite{TM-GB-DSB-FP:19a} and approximate
\cite{TM-GB-DSB-FP:19} stability conditions for cluster
synchronization, which are also used here. Compared to the above
references, this paper focuses on the control of cluster
synchronization, rather than on its enabling~conditions.

\noindent\textbf{Paper contribution.} The contributions of this paper
are twofold. On the technical side, we formulate and solve a network
optimization problem to enforce stable cluster synchronization among
interconnected Kuramoto oscillators (Section \ref{sec: section 3}). We
provide a two-step procedure to compute the smallest (as measured by
the Frobenius norm) perturbation of the network weights and the
oscillators' natural frequencies so as to achieve a desired and
arbitrary synchronization pattern. Notably, the proposed algorithm
allows for the modification of only a selected subset of the network
parameters, as typically constrained in applications. We also prove
that cluster synchronization is robust to parameter mismatches and
numerical inaccuracies, which complements the theoretical derivations
in \cite{TM-GB-DSB-FP:19a,TM-GB-DSB-FP:19}, and strengthen the
applicability of our control methods to work in practice.

On the application side, this work contains the first mathematically
rigorous and neurologically plausible framework to control functional
connectivity in the brain, and takes a significant step to fill the
gap between empirical studies on oscillatory neural activity
\cite{GD-VK-ARM-OS-RK:09,CDH-AZS-MP-MC-ECL:17,JC-EH-OS-GD:11} and the
recent technical body of work inspired by neural synchronization
\cite{CF-AC-FP:17, LT-CF-MI-DSB-FP:17,
  TM-GB-DSB-FP:19a,TM-GB-DSB-FP:19}. In Section \ref{sec: section 4},
we apply our control technique to an empirically-reconstructed
structural brain network, and validate our results by computing the
correlation of resting-state fMRI signals obtained through a realistic
hemodynamic~model. As a minor contribution, our work extends
\cite{JC-EH-OS-GD:11} by allowing heterogeneous Kuramoto~dynamics.

\noindent\textbf{Mathematical notation.} The sets $\real_{>0}$,
$\mathbb{S}^1$ and $\mathbb{T}^n$ denote the positive real numbers,
the unit circle, and the $n$-dimensional torus, respectively. We
represent the vector of all ones with $\1$. The Frobenius and
$\ell_{2}$ norms are denoted as $\| \cdot \|_{\text{F}}$ and
$\| \cdot \|$, respectively, and $A \circ B$ is the Hadamard product
between matrices $A$ and $B$. A \mbox{(block-)diagonal} matrix is
denoted by $\mathrm{(blk)diag}(\cdot)$. We let $\im = \sqrt{-1}$.  Let
$A\ge 0$ represent an element-wise inequality on the entries of $A$,
$A^+$ the element-wise nonnegative part of $A$, and $A\succ0$ a
positive definite matrix $A$. We let $\lambda_i(A)$ and $\sigma_i(A)$
denote the $i$-th eigenvalue and the $i$-th singular value of
$A \in \real^{n\times n}$, respectively, and
$\lambda_{\text{max}}(A) = \max_i |\lambda_i(A)|$ and
$\lambda_{\text{min}}(A) = \min_i |\lambda_i(A)|$. Finally, we let
$\overline{\lambda}(A) = \frac{1}{n}\sum_i\lambda_i(A)$ and
$\overline \sigma(A) = \frac{1}{n}\sum_i \sigma_i(A)$.

\section{Problem setup and preliminary notions}\label{sec: setup}
The aim of this work is to control network parameters so that groups
of brain regions exhibit a high degree of functional connectivity. In
this context, functional interactions are defined as the pairwise
correlation between hemodynamic signals recorded in two brain
regions. One model used to simulate such hemodynamic signals is
described by a set of nonlinear differential equations
\cite{DM-MGP-CDG-GLR-MC:07} that can be approximated in the frequency
domain as a linear low-pass filter \cite{JC-EH-OS-GD:11}. Because the
only input to such hemodynamic model is the oscillatory neural
activity, the formation of strongly (functionally) connected brain
regions can be promoted by controlling the synchronization level of
their neural dynamics. We follow \cite{JC-EH-OS-GD:11} to model such
neural dynamics with a sparse network of heterogeneous Kuramoto
oscillators that are connected to each other according to the
anatomical architecture of the human brain, more specifically known as
white matter tracts.\footnote{We assume that at each node of a
  structural brain network there exists a community of excitatory and
  inhibitory neurons whose dynamical state is in a regime of
  self-sustained oscillation. In other words, the neurons' firing
  rates delineate a limit cycle, and their dynamics can be
  approximated by a single variable, which is the angle (or phase) on
  this cycle.}  Ultimately, the problem of generating desired patterns
of functional connectivity reduces to the one of controlling cluster
synchronization in a network of heterogeneous Kuramoto oscillators.

To be precise, let $\mc G = (\mc V, \mc E)$ be a weighted digraph,
where $\mc V = \until{n}$ and $\mc E \subseteq \mc V \times \mc V$
represent the oscillators, or nodes, and their interconnection edges,
respectively. The $i$-th oscillator's dynamics reads as:
\begin{align}\label{eq: kuramoto}
  \dot \theta_i = \omega_i + \sum_{j \neq i} a_{ij} \sin(\theta_{j}-\theta_{i}),
\end{align}
where $\omega_i \in \real_{> 0}$ denotes the natural frequency of the
$i$-th oscillator, $\theta_i \in \mathbb{S}^1$ is its phase,
$a_{ij} \in \real_{> 0}$ is the weight of the edge $(j,i) \in \mc E$,
with $a_{ij} = 0$ when $(j,i) \not\in \mc E$, and $A = [a_{ij}]$ is
the weighted adjacency matrix of $\mc G$.

To characterize synchronized trajectories among subsets of
oscillators, let $\mc P = \{\mc P_1, \dots, \mc P_m\}$ be a nontrivial
partition of $\mc V$, where each cluster contains at least two
oscillators and its graph is strongly connected.\footnote{As the brain
  is densely connected \cite{DSB-PZ-JIG:18}, this assumption is not
  restrictive.}  We say that a network exhibits cluster
synchronization when the oscillators can be partitioned so that the
phases of the oscillators in each cluster evolve
identically. Formally, we define the \emph{cluster synchronization
  manifold} associated with the partition $\mc P$ as
\begin{align*}
  \mc S_{\mc P} = \setdef{\theta \in \mathbb{T}^n }{ \theta_i =
  \theta_j  \text{ for all } i,j \in \mc P_k, k = 1,\dots,m} .
\end{align*}
Then, the network is cluster-synchronized with partition $\mc P$ when
the phases of the oscillators belong to $\mc S_{\mc P}$ at all times.
Without loss of generality, the oscillators are labeled so that $\mc P_k = \{\sum_{\ell=1}^{k-1}|\mc P_\ell| + 1,\dots,\sum_{\ell=1}^k|\mc P_\ell|\}$, where $|\mc P_\ell|$ denotes the cardinality of the set $\mc P_\ell$.

Because our control framework 
leverages conditions for the invariance and stability of the cluster synchronization manifold to modify the network weights and oscillators' natural frequencies, we briefly recall useful preliminary results that have recently  been established in \cite{LT-CF-MI-DSB-FP:17,TM-GB-DSB-FP:19a,TM-GB-DSB-FP:19}. Specifically, given a desired network partition $\mc P = \{\mc P_1, \dots , \mc P_m\}$, invariance of $\mc S_{\mc P}$ is guaranteed by the following conditions:
\begin{itemize}
  \item[(C1)] The natural
  frequencies satisfy $\omega_i = \omega_j$ for
  every $i,j \in \mc P_k$ and $k~\in~\until{m}$. Equivalently, $B_\text{span}^\transpose \omega = 0$,
  \end{itemize}
  where $B_\text{span} \in \real^{|\mc V| \times |\bigcup_k \mc E_{\text{span},k}|}$ is the incidence matrix of $\bigcup_{k=1}^m\mc T_{k}$, with $\mc T_k = (\mc P_k, \mc E_{\text{span},k})$ being a spanning tree of the digraph $\mc G_k$ of the isolated cluster $\mc P_k$;

  \begin{itemize}
\item[(C2)]  The network weights satisfy
  $\bar{V}_{\mc P}^{\transpose}\bar A V_{\mc P}=0$,
  \end{itemize}
where $V_{\mc P}\in \real^{n\times m}$ is the characteristic matrix of the network defined as $V_{\mc P} = \left[{v_1}/{\|v_1\|},\dots, {v_m}/{\|v_m\|}\right]$, with
\begin{equation*}
v_i^\transpose =  [ \underbrace{0,  \dots,  0}_{\sum_{j=1}^{i-1}| \mc P_j|}, \underbrace{1,  \dots,  1}_{|\mc P_i|}, \underbrace{0,  \dots,  0}_{\sum_{j=1+1}^{n}| \mc P_j|} ], 
\end{equation*}
$\bar V_{\mc P}\in\real^{n \times (n-m)}$ is an orthonormal basis of the orthogonal subspace to the image of $V_{\mc P}$, and $\bar A = A - A \circ V_{\mc P} V_{\mc P}^\transpose$ is the matrix of inter-cluster connections only (see also \cite{LT-CF-MI-DSB-FP:17}). 

We assume that the \emph{isolated} clusters are locally stable:
\begin{itemize}
\item[(A$1$)] The dynamics \eqref{eq: kuramoto}, with $a_{ij}=0$ when $i,j$ belong to different clusters, converges exponentially fast to $\mc S_{\mc P}$.
\end{itemize}
Notice that Assumption (A$1$) is satisfied when $\mc G_{k}$ has symmetric weights and condition (C$1$) holds \cite[Lemma 3.1]{TM-GB-DSB-FP:19a}\cite[Theorem 5.1]{Doerfler2014}. In our case, (A$1$) is not restrictive because structural brain networks are typically symmetric \cite{SG-FP-MC-QKT-BYA-AEK-JDM-JMV-MBM-STG-DSB:15}.\footnote{In a general case, one can ensure that Assumption (A$1$)  is satisfied simply by pairing the control mechanism developed in the next session with an independent one that makes intra-cluster connections symmetric.}

  Let $\omega^{(k \ell)}$ denote the natural frequency difference between any two nodes in disjoint clusters $\mc P_k$ and $\mc P_\ell$. If (C$1$) and (C$2$) hold, then a tight approximate condition for $\mc S_{\mc P}$ to be locally exponentially stable is \cite{TM-GB-DSB-FP:19}:
  \begin{itemize}
  \item[(C3)] The natural frequencies and the network weights are such that $\lambda_{\text{max}}(\Xi(A,\omega)) < 1$, with $\Xi = [\xi_{k\ell}]$ and
\end{itemize}
  \begin{align}\label{eq:approx gains av}
    \xi_{k\ell} \!= \!
    \begin{cases} 
      \nu_{k\ell}\overline{\sigma}(G_{k}(\im \omega^{(k\ell)})), &
      \hspace{-0.2cm}\text{if }  \overline{\lambda}(J_{k}) \le \overline{\lambda}(J_{\ell}),\\[.35em]
      \nu_{k\ell}\frac{
        \overline{\sigma}(G_{k}(0))}{\overline{\sigma} (G_{\ell}(0))}
      \overline{\sigma}(G_{\ell}(\im \omega^{(k\ell)})), &\hspace{-0.2cm}\text{if }
      \overline{\lambda}(J_{\ell}) \!< \overline{\lambda}(J_{k}),
    \end{cases}
  \end{align}
  where $G_{k}(s)=(sI-J_{k})^{-1}$, $J_k$ is the Hurwitz stable Jacobian matrix of the intra-cluster phase difference dynamics and $\nu_{k\ell}$ is a function of the inter-cluster weights.
Due to space constraints, we refer the interested reader to \cite{LT-CF-MI-DSB-FP:17,TM-GB-DSB-FP:19} for a detailed discussion on conditions (C1), (C2), and (C3).

\section{Control of cluster~synchronization}\label{sec: section 3}

In this section, we propose a control mechanism to obtain a prescribed and robust configuration of synchronized oscillatory patterns.
Towards this aim, we consider a network $\mc G = (\mc V, \mc E)$ and an arbitrary partition $\mc P = \{\mc P_1, \dots, \mc P_m\}$ of $\mc V$. The proposed control technique is minimally invasive in the sense that it looks for the smallest correction (in the Frobenius norm sense) of inter-cluster network weights and oscillators' natural frequencies that renders the cluster synchronization manifold $\mc S_{\mc P}$ invariant and locally stable. In practice, a modification of the network parameters will require either the exploitation of neural plasticity or localized surgical intervention for the modification of the network weights and structure, and pharmacological or electromagnetic influence for the refinement of the brain regions' natural frequencies.
In mathematical terms, the approach is encoded into solving the following constrained minimization problem:
\begin{align}\label{eq: control complete}
      \min\limits_{\Delta,\mu} \ &  \left\|[\Delta, \,  \mu]\right\|_{\text{F}}^2\\[.3em]
      \text{s.t.} \ \  &\bar{V}_{\mc P}^{\transpose}(\bar A + \Delta) V_{\mc P}=0, \label{eq: control complete a} \tag{\theequation a}\\[.3em]
      & B_{\text{span}}^{\transpose}(\omega +\mu) = 0, \label{eq: control complete b} \tag{\theequation b}\\[.3em]
      & H^\text{c} \circ \Delta=0 \label{eq: control complete c} \tag{\theequation c}\\[.3em]
      & \bar A+\Delta\ge 0, \label{eq: control complete d} \tag{\theequation d} \\[.3em]
      & \omega +\mu \ge 0, \label{eq: control complete d2} \tag{\theequation e} \\[.3em]
      & \lambda_{\text{max}}(\Xi(A+\Delta,\omega+\mu)) < 1, \label{eq: control complete e} \tag{\theequation f}
  \end{align}
  where $\Delta$ is the correction of the network matrix, $\mu$ is the
  correction of the natural frequencies vector, and the $(i,j)$-th
  entry of $\Delta$ is zero if $i$, $j$ belong to the same partition
  $\mc P_{k}$,~$k\in\until{m}$. Further, $H$ is the $0$-$1$ adjacency
  matrix of $\mc H = (\mc V, \mc E_{\mc H})$, which is the graph
  encoding the set of edges $\mc E_{\mc H}\subseteq \mc E$ that is
  allowed to be modified, and $H^{\text{c}} =
  \1\1^{\transpose}-H$. That is, the $(i,j)$-th entry of a solution
  $\Delta^*$ to problem \eqref{eq: control complete} is zero when the
  corresponding $(i,j)$-th entry of $H$ is zero. The optimization
  problem \eqref{eq: control complete} is illustrated in
  Fig.~\ref{fig: optimization problem}.
 
 \begin{figure}[t]
  \centering
    \includegraphics[width=1\columnwidth]{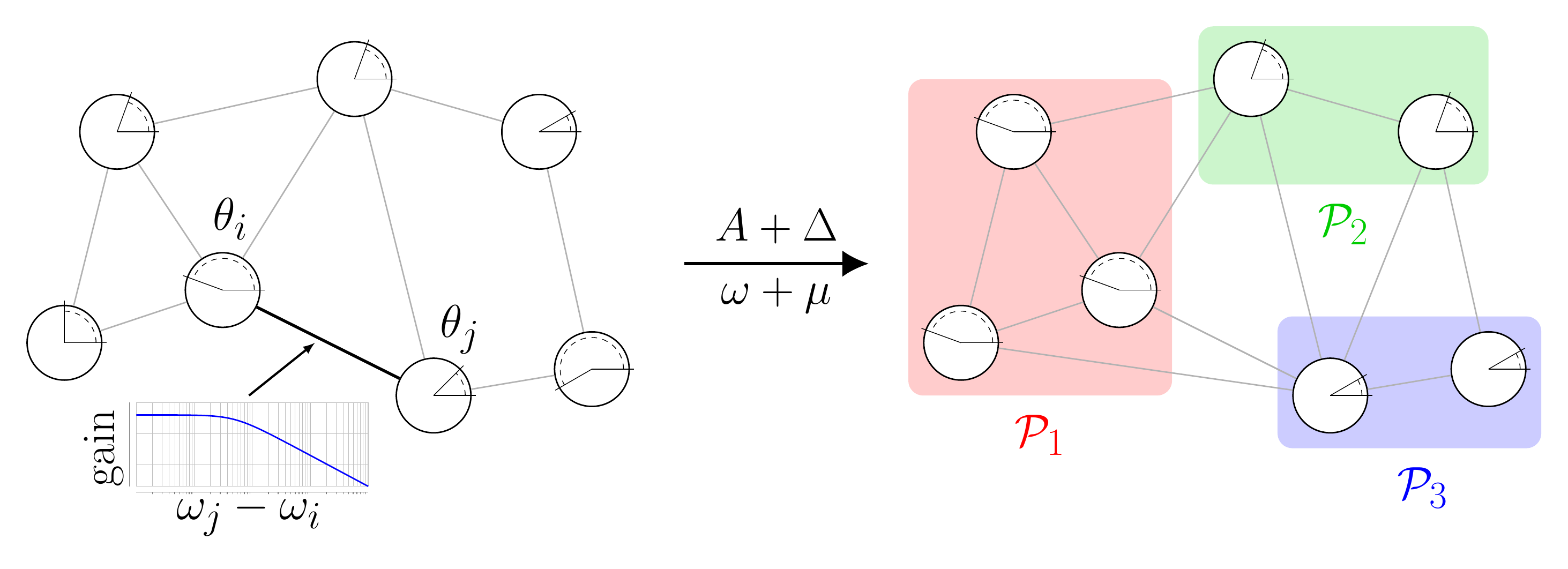}
    \caption{The left depicts a network of oscillators. The coupling
      strength between the oscillators depends on the network weights
      and the differences of their natural frequencies
      \cite{TM-GB-DSB-FP:19}. The optimization problem \eqref{eq:
        control complete} seeks for the smallest modification of the
      network weights and the oscillators' natural frequencies to
      ensure a desired stable pattern of cluster synchronization
      (right panel). We remark that the techniques used in this paper
      for cluster synchronization in frequency-weighted networks of
      Kuramoto oscillators, are applicable to a broad class of network
      optimization problems, e.g., \cite{TM-VK-DSB-FP:18}.}
  \label{fig: optimization problem}
\end{figure}
 
Constraints \eqref{eq: control complete a} and \eqref{eq: control complete b} are equivalent to conditions (C$2$) and (C$1$), respectively, for the invariance of $\mc S_{\mc P}$. Constraint \eqref{eq: control complete c} restricts the corrective action to a subset of all the possible interconnections, in affinity with the practical limitations of localized interventions. Constraints \eqref{eq: control complete d} and \eqref{eq: control complete d2} are due to biological compatibility and require the inter-cluster weights of the perturbed network and oscillators' natural frequencies to be nonnegative. Finally, Constraint \eqref{eq: control complete e} corresponds to (C3) and guarantees the (local) stability of $\mc S_{\mc P}$. In particular, the latter constraint makes the above problem non-convex and, therefore, potentially intractable from a numerical viewpoint. To overcome this issue, we next propose a suboptimal, yet numerically more tractable, control strategy. 
Specifically, we decouple \eqref{eq: control complete} into two simpler subproblems. The first one solves for the smallest correction of  inter-cluster weights satisfying \eqref{eq: control complete a}, \eqref{eq: control complete c}, and \eqref{eq: control complete d}, whereas the second one solves for the smallest correction of the oscillators' natural frequencies satisfying  \eqref{eq: control complete b},  \eqref{eq: control complete d2} and~\eqref{eq: control complete e}.
 
\subsection{Inter-cluster structural control for invariance of $\mc S_{\mc P}$}

We first address the problem of computing the smallest correction of inter-cluster weights such that constraints \eqref{eq: control complete a}, \eqref{eq: control complete c}, and \eqref{eq: control complete d} are satisfied. Specifically, we focus on the following minimization problem:
\begin{align}
      \min\limits_{\Delta} \ &  \|\Delta\|_{\text{F}}^{2} \label{eq: control invariance} \\[.3em]
      \text{s.t.} \ \  &\bar{V}_{\mc P}^{\transpose}(\bar A + \Delta) V_{\mc P}=0,\tag{\theequation a} \label{eq: control invariance a}\\[.3em]
      & H^\text{c} \circ \Delta=0,\tag{\theequation b} \label{eq: control invariance b} \\[.3em]
      & \bar A+\Delta\ge 0. \tag{\theequation c} \label{eq: control invariance c}
  \end{align}
The optimization problem \eqref{eq: control invariance} is convex and, when feasible, it can be efficiently solved by means of standard optimization techniques. Feasibility of \eqref{eq: control invariance} depends on the constraint graph $\mc H$ (see Remark \ref{remark: additional constraints}).
In what follows, we present a simple and efficient projection-based algorithm to solve this~problem.

\begin{theorem}{\bfseries\em (Smallest sparse inter-cluster correction)} \label{thm correction}
Assume that the problem \eqref{eq: control invariance} is feasible, and consider the matrix sequence $\{Z_{k}\}_{k\ge 0}$ generated via the following iterative procedure:
\begin{equation}  \label{eq alt proj}
\begin{aligned}
Y_{k}  &= H\circ (Z_{k}+T_{k})^{+} +H^{\text{c}}\circ \bar A,  \\
T_{k+1} & = Z_{k} +T_{k} -Y_{k}, \\
Z_{k+1} &  = Y_{k} + Q_{k} -\bar V_{\mc P} \bar V_{\mc P}^{\transpose} (Y_{k}+Q_{k}) V_{\mc P} V_{\mc P}^{\transpose},\\
Q_{k+1} & = Y_{k} +Q_{k} - Z_{k+1},
\end{aligned}
\end{equation}
where $Z_{0}=\bar{A}$, and $T_{0}=Q_{0}=0$.
Then, the sequence $\{Z_{k}\}_{k\ge 0}$ converges to a matrix $Z^{*}$, and a minimizer of \eqref{eq: control invariance} subject to \eqref{eq: control invariance a}, \eqref{eq: control invariance b}, and \eqref{eq: control invariance c}, has the form $\Delta^{*} = Z^{*}-\bar{A}$.
  \end{theorem}
  \smallskip
  
  \begin{proof}
  Let $\Pi_{\mc Z}(W) =\arg\min_{Z\in\mc Z} \|Z-W\|_{\text{F}}^{2}$ denote the projection (in the Frobenius norm sense) of $W$ onto a convex set $\mc Z$, and define the closed convex sets $\mc Z_{1}=\setdef{Z\in\real^{n\times n}}{ Z\ge 0 \text{ and } H^{\text{c}}\circ Z =\bar A}$ and $\mc Z_{2}=\setdef{Z\in\real^{n\times n}}{\bar{V}_{\mc P}^{\transpose}Z V_{\mc P}=0}$.
  Note that $\Pi_{\mc Z_{1}}(W)= H\circ W^+ + H^\text{c}\circ \bar A$ and, 
by Lemma \ref{lemma optimal Delta symm} in the Appendix,%
\begin{align*}
\Pi_{\mc Z_{2}}(W) & = \arg\min_{Z\in\mc Z_{2}} \|Z-W\|_{\text{F}}^{2}\\
& = W -\bar V_{\mc P} \bar V_{\mc P}^{\transpose} W V_{\mc P} V_{\mc P}^{\transpose}, 
\end{align*}
for any $W$.
Hence, the sequence $\{Z_{k}\}_{k\ge 0}$ generated by \eqref{eq alt proj} coincides with the sequence generated by Dykstra's projection algorithm \cite{JPB-RLD:86} applied to the projections onto $\mc Z_{1}$ and $\mc Z_{2}$.  Since the problem \eqref{eq: control invariance} is feasible, $\mc Z_{1} \cap \mc Z_{2}\ne \emptyset$, and the latter sequence converges to a matrix $Z^{*}=\Pi_{\mc Z_{1}\cap \mc Z_{2}}(\bar A) = \arg\min_{Z \in\mc Z_{1}\cap \mc Z_{2}}\|Z-\bar A\|_{\text{F}}^{2}$  \cite{JPB-RLD:86}. Finally,\begin{align*}
Z^{*} & = \arg\min_{Z \in\mc Z_{1}\cap \mc Z_{2}}\|Z-\bar A\|_{\text{F}}^{2}\\
& = \bar A +\arg\min_{\substack{\Delta \text{ s.t. } \eqref{eq: control invariance a},\, \eqref{eq: control invariance b},\, \eqref{eq: control invariance c}}} \|\Delta\|_{\text{F}}^{2},
\end{align*}
and the statement follows.
  \end{proof}
  
\smallskip
  
  \begin{remark}{\bfseries\em (Sufficient condition for the feasibility of \eqref{eq: control invariance})}\label{remark: additional constraints} 
Recall from \cite{LT-CF-MI-DSB-FP:17} that condition (C$2$) is equivalent to: 
  \begin{equation}\label{eq: lorenzo}\sum_{k \in \mc P_\ell} a_{i k} - a_{j k} = 0
  \end{equation} for every $i,j \in \mc P_z$ and for all $z,\ell \in \until{m}$, with $z \neq \ell$. 
  Notice that, if for every $i \in \mc P_z$ there exists at least one $(k,i)\in\mc E_{\mc H}$, $k \in \mc P_\ell$,  a solution to \eqref{eq: lorenzo} can always be found and problem \eqref{eq: control invariance} is feasible.
 \oprocend
  \end{remark}
  
\subsection{Frequency tuning for local stability of $\mc S_{\mc P}$}\label{sec: frequency tuning}
We now turn to the problem of computing the smallest correction of natural frequencies such that constraints \eqref{eq: control complete b}, \eqref{eq: control complete d2}, and \eqref{eq: control complete e} are satisfied. That is,
\begin{align}\label{eq: control stability}
      \min\limits_{\mu} \ &  \|\mu\|_\text{F}^{2} \\[.3em]
      \text{s.t.} \ \
      & B_{\text{span}}^{\transpose}(\omega +\mu) = 0, \tag{\theequation a} \label{eq: control stability a} \\[.3em]
      & \omega +\mu \ge 0, \label{eq: control stability b} \tag{\theequation b} \\[.3em]
     &  \lambda_{\text{max}}(\Xi(A,\omega+\mu)) < 1. \tag{\theequation c} \label{eq: control stability c}
  \end{align}

\begin{theorem}{\bfseries\em (Feasibility of problem \eqref{eq: control stability})}\label{thm frequency tuning}
There always exists a correction $\mu$ satisfying \eqref{eq: control stability a}, \eqref{eq: control stability b}, and \eqref{eq: control stability c}.
\end{theorem}
\begin{proof}
Consider the vector $\mu=[\mu_{1},\dots,\mu_{n}]^\transpose$. Note that we can find some $\mu_{i}$ such that (i) $\omega_{i}+\mu_{i}=\omega_{j}+\mu_{j}>0$ for all $i,j \in \mc P_k$, $k\in\until{m}$, and  (ii) $|\omega_{i}+\mu_{i}-(\omega_{j}+\mu_{j})|>\eta$, for  all $i \in \mc P_k$, $j\in\mc P_{\ell}$, $k,\ell\in\until{m}$, $k\ne \ell$, and $\eta>0$ arbitrarily large. From (i), $\mu$ satisfies \eqref{eq: control stability a} and \eqref{eq: control stability b}. Further, since each nonzero entry of $\Xi(A,\omega+\mu)$ in \eqref{eq:approx gains av} behaves as a low-pass filter, by fact (ii) $\lambda_{\text{max}}(\Xi(A,\omega+\mu))$ can be made arbitrarily small. This implies that there always exists a vector $\mu$ satisfying \eqref{eq: control stability c} and concludes the proof.
\end{proof}

An optimal solution to \eqref{eq: control stability} is typically difficult to compute, because of the eigenvalue constraint \eqref{eq: control stability c}. However, several heuristics can be used to compute a suboptimal correction in \eqref{eq: control stability}. For instance, we next outline an effective procedure to find a suboptimal solution to \eqref{eq: control stability}. Let $\omega_\text{av}^{(k)}=\frac{1}{|\mc P_{k}|}\sum_{i\in \mc P_{k}} \omega_{i}$ denote the average frequency within each cluster, and let
\begin{align*}
\omega_\text{av}&=[\omega_{\text{av},1},\dots,\omega_{\text{av},n}]^\transpose\\[.3em]
&=[\underbrace{\omega_{\text{av}}^{(1)},\dots,\omega_{\text{av}}^{(1)}}_{|\mc P_1|},\dots,\underbrace{\omega_{\text{av}}^{(m)},\dots,\omega_{\text{av}}^{(m)}}_{|\mc P_m|}]^{\transpose}.\end{align*} Further, define the quotient graph $\mc Q = (\mc V', \mc E')$ where each node in $\mc V'$ represents a cluster and each edge in $\mc E'$ an interconnection between two clusters. Our procedure leverages Theorem \ref{thm frequency tuning} and increases the frequency differences between pairs of connected clusters until constraint \eqref{eq: control stability c} is satisfied.
The procedure consists of four steps:
\begin{enumerate}
\item  If $\lambda_{\text{max}}(\Xi(A,\omega_\text{av}))<1$, then $\mu^{*} =[\mu_{1}^{*},\dots,\mu_{n}^{*}]$, with $\mu_{i}^{*} = \omega_{\text{av},i}-\omega_{i}$ is an optimal correction to \eqref{eq: control stability}. Otherwise, proceed to the next step.
\item Construct a depth-first spanning tree $\mc T_{\mc Q}$ of $\mc Q$ rooted at $r =\arg\min_{k} \omega_{\text{av}}^{(k)}$.\footnote{Notice that such a spanning tree always exists, since $\mc Q$ is~connected.}
\item Assign the frequency $\omega(k,\alpha)=\omega^{(r)} + k\alpha$, $\alpha>0$, to each node of each cluster in $\mc T_{\mc Q}$ of depth $k$,  $k=1,2,\dots,k_{\text{max}}$, where $k_{\text{max}}$ denotes the height of~$\mc T_{\mc Q}$.\footnote{Given a connected graph $\mc G =(\mc V, \mc E)$ and a spanning tree $\mc T$ of $\mc G$ rooted at $r\in\mc V$, the depth of a node $v\in\mc V$ is the length of the path in $\mc T$ from $r$ to $v$, and the height of $\mc T$ is  the maximum depth among the nodes in $\mc V$.} Let $\omega(\alpha)=[\omega_{1}(\alpha),\dots,\omega_{n}(\alpha)]^\transpose$ denote the resulting vector of modified frequencies.
\item Find the smallest $\alpha^*$ satisfying $\lambda_{\text{max}}(A,\omega(\alpha^{*}))<1$. Then, $\mu^{*} =[\mu_{1}^{*},\dots,\mu_{n}^{*}]$, with $\mu_{i}^{*} = \omega_{i}(\alpha^{*})-\omega_{i}$, is a (suboptimal) solution to \eqref{eq: control stability}.
\end{enumerate}

\section{Robustness of the control framework}\label{sec: robustness}

In this section, we show that the control framework described in Section \ref{sec: section 3}, and in fact the stability property of the cluster synchronization manifold $\mc S_{\mc P}$, is robust to perturbations of the network parameters. That is, small changes in the oscillators' natural frequencies and network weights yield a small deviation from cluster-synchronized trajectories.
In light of this, the proposed control mechanism lends itself to practical applications, where the network parameters are not known exactly and the neural dynamics is subject to noise.

Consider the dynamics \eqref{eq: kuramoto} with perturbed parameters:
\begin{align}\label{eq: theta perturbed}
\dot{\theta}_{i} = \tilde\omega_{i} +\sum_{j\ne i} \tilde a_{ij}\sin(\theta_{j}-\theta_{i}),
\end{align} 
where $\tilde\omega_{i}=\omega_{i}+\delta \omega_{i}$ and $\tilde a_{ij}=a_{ij}+\delta a_{ij}$. 
Notice that, if $\delta \omega_{i}=0$ and $\delta a_{ij}=0$, the dynamics \eqref{eq: theta perturbed} is equivalent to~\eqref{eq: kuramoto}.
From \eqref{eq: theta perturbed}, the perturbed intra-cluster difference dynamics of nodes $i,j\in\mc P_k$, with $k\in\until{m}$, reads as:
\begin{align}\label{eq: perturbed differences dynamics}
  &\dot \theta_j - \dot \theta_i =\,  \omega_j + \delta\omega_i -\omega_i-\delta\omega_j \notag\\
  &\hspace{-0.1cm}+\!\sum_{z=1}^n\! \left[(a_{jz}+\delta a_{jz})\sin(\theta_z-\theta_j)\!-\!(a_{iz}+\delta a_{iz})\sin(\theta_z-\theta_i)\right]\!\notag\\
  &\hspace{-0.1cm}=  \omega_j - \omega_i + \sum_{z=1}^n \left[ a_{jz}\sin(\theta_z-\theta_j)- a_{iz}\sin(\theta_z-\theta_i)\right] + \delta_{ij},
\end{align}
where $\delta_{ij} = \delta\omega_j - \delta\omega_i + \sum_{z=1}^n [\delta a_{jz}\sin(\theta_z-\theta_j)$ $-\delta a_{iz}\sin(\theta_z-\theta_i)]$.
 Finally, let $\delta$ be the vector of all $\delta_{ij}$ that affect the nominal intra-cluster dynamics as in \eqref{eq: perturbed differences dynamics}.

 We are now ready to present the main result of this section, which
 resorts to the prescriptive stability condition derived in
 \cite{TM-GB-DSB-FP:19a} that we recall in the Appendix \ref{sec: AB} for
 completeness.
\begin{theorem}{\bf \emph{(Robustness of cluster synchronization)}}\label{thm: perturb cluster sync} Assume
  that the network weights satisfy Theorem \ref{thm: stability of Sp}, and consider any pair of nodes
  $i,j\in\mc P_k$, $k\in\until{m}$. Then, for some finite $T>0$ and
  for all initial conditions such that
  $|\theta_j(0)-\theta_i(0)|<\varepsilon$, with $\varepsilon>0$
  sufficiently small, the solution to the perturbed dynamics
  \eqref{eq: theta perturbed} satisfies
\begin{equation}\label{eq: bound c}
|\theta_j(t)-\theta_i(t)| \le c\,\gamma\ \quad \forall t\ge T,
\end{equation}
where $\gamma = \max_{\theta\in[0,2\pi)}\|\delta\|$, and $c$ is a
constant that depends only on the network weights.
\end{theorem}
\smallskip

\begin{proof}
In the first part of the proof, we combine the Lyapunov functions for the isolated clusters $\mc P_k$, $k = 1,\dots, m$, into a Lyapunov function for the intra-cluster differences dynamics of the whole network. In the second part of the proof, we show that such Lyapunov function satisfies certain bounds, so that the application of \cite[Lemma 9.2]{HKK:02} suffices to prove the claimed statement. 

We let $x_{ij} = \theta_j-\theta_i$, and $S$, $x_\text{intra}$ and $J_k$ be as in the Appendix \ref{sec: AB}.
To combine the Lyapunov functions of the isolated clusters, we note that if $S$ is an $M$-matrix, then, along the lines of \cite[Proof of Theorem 3.2]{TM-GB-DSB-FP:19a}, 
the origin of the nominal intra-cluster dynamics of $x_\text{intra}$ is locally exponentially stable with Lyapunov function 
\begin{equation}\label{eq: V nom sys}
V(x_\text{intra}) = \sum_{k=1}^m d_k x_\text{intra}^{(k)\transpose} P_k x_\text{intra}^{(k)},
\end{equation}
where $P_k\succ 0$ satisfies $J_k^\transpose P_k + P_kJ_k = -I$, and  $d_k>0$ are such that $DS+S^\transpose D \succ 0$, with
$
D = \mathrm{diag}(d_1 \dots, d_m)
$~\cite{HKK:02}.

Consider now the Lyapunov function \eqref{eq: V nom sys}, and notice that:
$
c_1 \| x_\text{intra} \|^2 \le \sum_{k=1}^m d_k x_\text{intra}^{(k)\transpose} P_k x_\text{intra}^{(k)} \le c_2\| x_\text{intra} \|^2,
$
with $c_1 = m d_\text{min} \min_k\lambda_\text{min}(P_k)$ and $c_2 = m d_\text{max}\max_k \lambda_\text{max}(P_k)$. Further, in the ball of radius $r$ of the origin  $\mc B_r = \setdef{x_\text{intra}}{\|x_\text{intra}\|<r, \, \dot V(x_\text{intra})<0}$, it holds that $
\dot V(x_\text{intra}) \le -c_3 \| x_\text{intra} \|^2$,
with $c_3=\lambda_\text{min}(DS+S^\transpose D)/2$. To see this, consider the derivative of the Lyapunov function $V(x_\text{intra})$ along the trajectories of the nominal system. Then, from \cite[\S9.5]{HKK:02} and \cite[proof of Theorem 3.2]{TM-GB-DSB-FP:19a}, the following inequality holds in $\mc B_r$:
\begin{align*}
\dot V &\le - \frac{1}{2} \begin{bmatrix} \|x_\text{intra}^{(1)}\|& \dots & \|x_\text{intra}^{(m)}\| \end{bmatrix}(DS+S^\transpose D)\begin{bmatrix}\|x_\text{intra}^{(1)}\|\\ \vdots \\ \|x_\text{intra}^{(m)}\|\end{bmatrix}\\
&\le - \frac{1}{2} \lambda_\text{min}(DS+S^\transpose D)\left\|x_\text{intra} \right\|^2,
\end{align*}
and $c_3$ follows. Further, since $\|\partial V / \partial x_\text{intra} \| = \| 2 x_\text{intra}^\transpose P_k \| < 2 \lambda_\text{max}(P_k)\|x_\text{intra} \|$, we have $
\|\partial V / \partial x_\text{intra}\| \le c_4  \| x_\text{intra} \|
$, with $c_4 = 2c_2$. Finally, once the constants $c_1, c_2, c_3$, and $c_4$ are computed, the definition of $x_\text{intra}$ and \cite[Lemma 9.2]{HKK:02} conclude the proof.
\end{proof}
\smallskip

Importantly, Theorem \ref{thm: perturb cluster sync} can be used to provide a quantitative bound on the asymptotic value of $|\theta_j - \theta_i|$. In fact, we can compute the constant $c$ in \eqref{eq: bound c} by exploiting \cite[Lemma 9.2]{HKK:02} and $c_1, c_2, c_3,c_4$ derived in the above proof.

\section{Control of functional connectivity in an empirically-reconstructed brain network}\label{sec: section 4}

We conclude this paper with the application of the control mechanism
presented in Section~\ref{sec: section 3} to the brain network
estimated in \cite{HP-CL-GX-MR-HCJ-WVJ-SO:08}, which is publicly
available at \url{http://umcd.humanconnectomeproject.org/umcd}. In
these data, structural connectivity is proportional to large-scale
connection pathways between cortical regions, and the gray matter is
subdivided into $n = 66$ cortical regions ($33$ per hemisphere). To
show the effectiveness of our proposed method in enforcing desired
functional connectivity by means of arbitrary synchronization
patterns, we partition the structural brain network in $3$ clusters,
i.e. $\mc P= \{\mc P_1, \mc P_2, \mc P_3\}$, each one comprising $22$
regions that do not belong to any known functionally connected
resting-state network. The three clusters are highlighted with
different colors in Fig.~\ref{fig: SC BrainNet} and Fig.~\ref{fig: A}.
Furthermore, following our goal of providing a method to enhance the
synchronization properties of a diseased or damaged brain, we simulate
the effects of brain damage, e.g., a stroke, by damping the
connectivity of one cluster \cite{AMT-LS-ES-MR-GF-DGN-FL:13}.  That
is, we weaken the intra-cluster connections of the first cluster by a
scaling factor $10^{-2}$ to echo reduced structural connectivity, and
we show that our technique can in fact recover stability of the
cluster synchronization manifold associated with the desired network
partition.\footnote{Specifically, weakening the intra-cluster
  connections of one cluster is likely to make $\mc S_{\mc P}$
  unstable \cite{TM-GB-DSB-FP:19}.}

Before presenting our results, we describe the methodology used to simulate human rs-fMRI functional connectivity.

\begin{figure}[t]
  \centering
  \subfigure[]
  {
    \includegraphics[width=0.44\columnwidth]{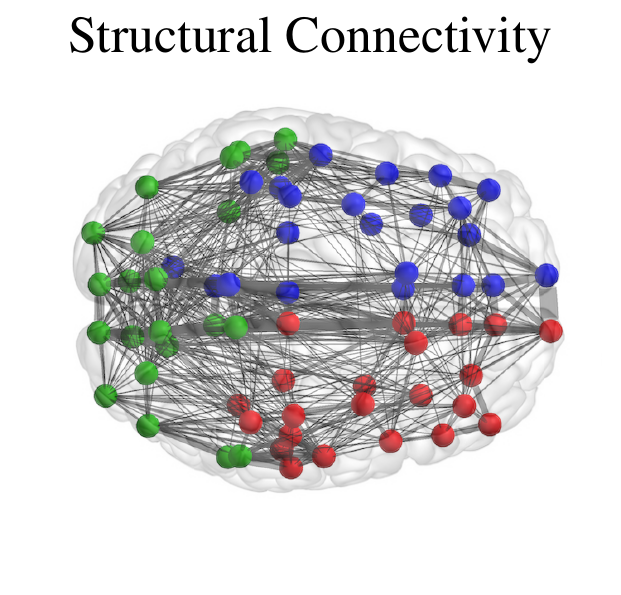}
    \label{fig: SC BrainNet}
    }
          \hspace{0.165cm}\subfigure[]
      {
    \includegraphics[width=0.44\columnwidth]{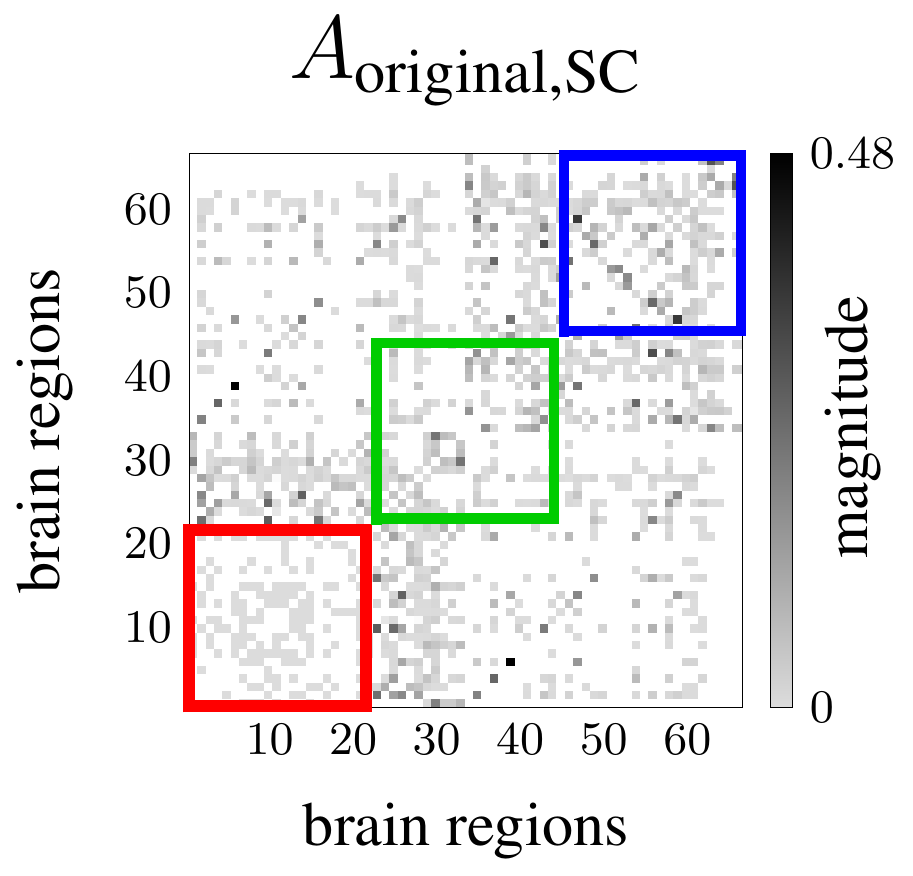}
    \label{fig: A}
    }\vspace{-0.1cm}
    \subfigure[]
       {
    \includegraphics[width=0.46\columnwidth]{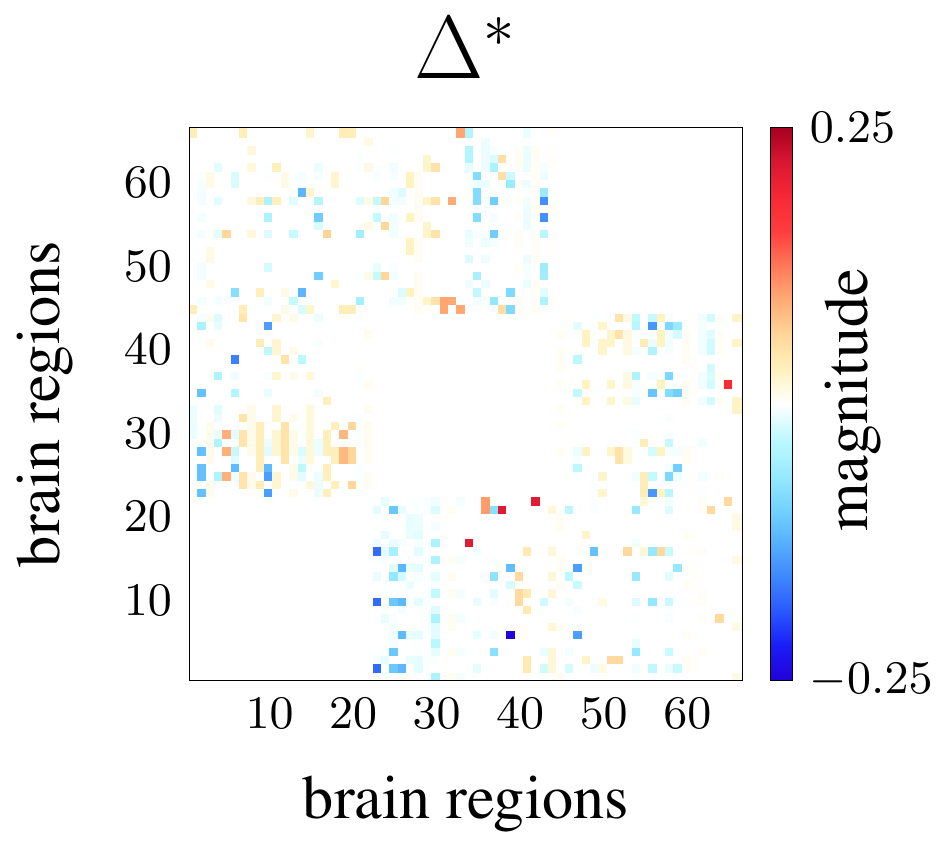}
    \label{fig: Delta}
    }
      \subfigure[]
      {
    \includegraphics[width=0.44\columnwidth]{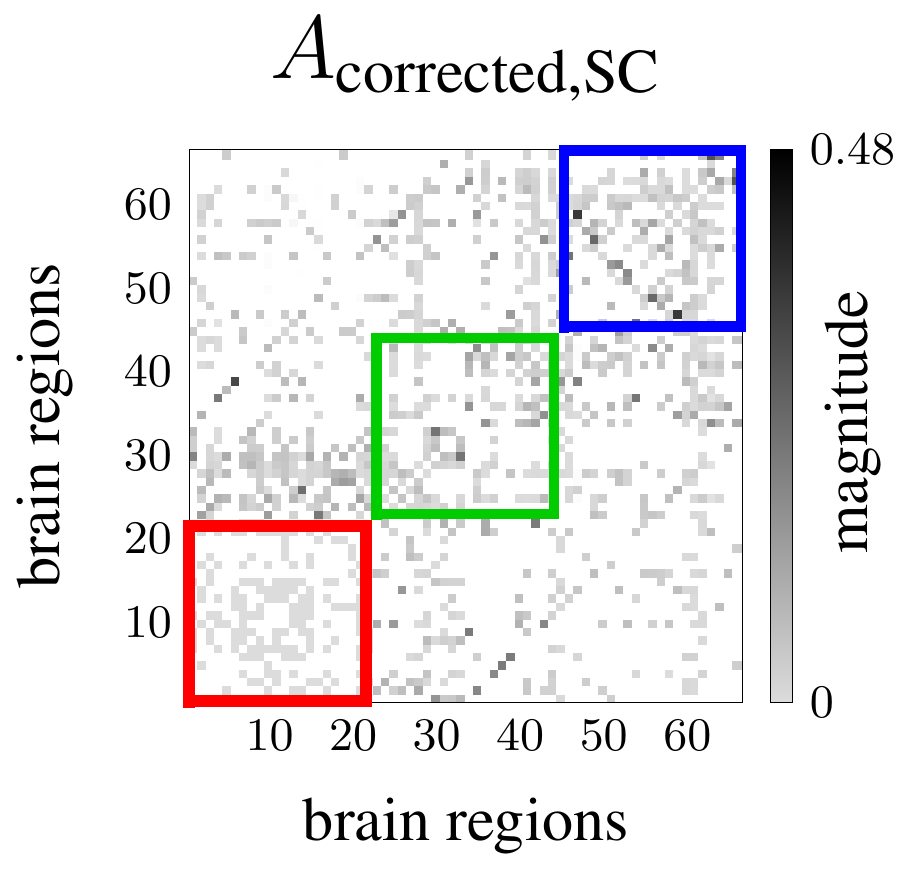}
    \label{fig: A+Delta}
    }
  \caption{Fig.~\ref{fig: SC BrainNet} depicts an axial view of the structural connectivity estimated in \cite{HP-CL-GX-MR-HCJ-WVJ-SO:08}, and was obtained with \emph{BrainNet
      Viewer} \cite{MX-JW-YH:13}. The edge thickness is proportional to the number of white matter streamlines connecting different regions. Fig.~\ref{fig: A} represents the adjacency matrix of the structural brain network in Fig.~\ref{fig: SC BrainNet}, where the white entries correspond to zero, and the intra-cluster connections in the first cluster (red nodes in Fig.~\ref{fig: SC BrainNet}) have been weakened to simulate the effect of brain damage. Fig.~\ref{fig: Delta} represents the network matrix correction $\Delta^*$ solution to the iterative procedure \eqref{eq alt proj} in Theorem \ref{thm correction}. Finally,  Fig.~\ref{fig: A+Delta} represents the matrix $A_{\text{corrected,SC}} = A_{\text{original,SC}} + \Delta^*$, where the total change of the edge weights amounts to $17\%$ (in the Frobenius norm) of the original~ones.}\vspace{-0.1cm}
  \label{fig: matrices corrections}
\end{figure}

\begin{figure*}[t]
  \centering
    \includegraphics[width=2\columnwidth]{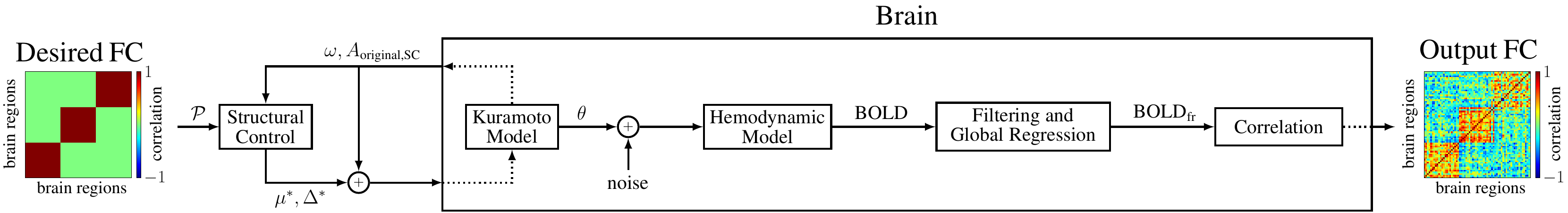}\vspace{-0.05cm}
  \caption{Schematic illustrating the pipeline to obtain desired functional connectivity (FC) from a structural connectivity matrix $A_\text{original,SC}$, the original natural frequencies $\omega$, and a desired network partition $\mc P$. The addition of noise to the synchronized neural dynamics $\theta$ represents the presence of background noise. The output matrix depicts the brain regions' functional connectivity simulated by computing the correlations of filtered and regressed BOLD signals.}\vspace{-0.1cm}
  \label{fig: pipeline}
\end{figure*}

\subsection{Simulation of functional connectivity}

The brain's neural activity is simulated through a network of coupled
Kuramoto oscillators, where we randomly draw the natural frequencies
of each oscillator from a uniform distribution in the range $[0, 60]$
$[$Hz$]$ so as to include all meaningful neural frequency bands
\cite{DM-MGP-CDG-GLR-MC:07}. We set the initial phases in the interval
$[0, 0.5]$ $[$rad$]$.
The Kuramoto phases act as an input to the neurovascular coupling,
which is modeled by the Balloon-Windkessel hemodynamic process
\cite{KJF-AM-RT-CKP:00}, and whose output is the blood-oxygen-level
dependent (BOLD) signal that is measured by rs-fMRI.

The neuronal activity $z_i$ of the $i$-th brain region produces an increase in a vasodilatory signal $s_i$, which is subject to auto-regulatory feedback. The inflow $f_i$ responds in proportion to this signal with concomitant changes in blood volume $\mu_i$ and deoxyhemoglobin content~$q_i$. Mathematically, the dynamics of these quantities reads as:
\begin{align*}
&\dot s_i = z_i - \kappa_i s_i - \gamma_i(f_i-1), \quad \dot f_i = s_i, \\
&\tau \dot \mu_i = f_i - \mu_i^{1/\alpha}, \quad \tau_i\dot q_i = f_i E(f_i,\rho_i)/\rho_i - \mu_i^{1/\alpha}q_i/\mu_i.
\end{align*}
The oxygen extraction is a function of the flow $E(f,\rho) = 1-(1-\rho^{\/f})$ where $\rho$ denotes the resting oxygen extraction fraction. The biophysical parameters $\kappa, \gamma, \tau, \alpha,$ and $\rho$ are exhaustively treated in \cite{KJF-AM-RT-CKP:00}. Finally, the BOLD signal is described as a static nonlinear function:
\begin{equation*}
y_i = V_0(k_1(1-q_i) + k_2(1-q_i/v_i) + k_3(1-\mu_i)),
\end{equation*}
where $V_0 = 0.02$ denotes the resting blood volume fraction, and
$k_1 = 7\rho_i$, $k_2 = 2$, $k_3 = 2\rho_i - 0.2$. Following
\cite{JC-EH-OS-GD:11}, we choose $z_i = \sin(\theta_i)$. Further, to
account for the presence of background noise in the brain, we add
white noise to the neural activity $z_i$ with variance $10^{-2}$. We
simulate $2$ minutes of BOLD signals and process the timeseries as
explained below in order to compute functional connectivity estimates
that closely resemble that of human rs-fMRI recordings.

To reduce the effect of spurious correlations from small and non-physiological high-frequency components, we filter the synthetic BOLD signals through a low-pass filter. Consequently, to improve the correspondence between resting-state correlations and anatomical connectivity, we process all of the simulated regional BOLD signals by a global signal regression \cite{MDF-DZ-AZS-MER:09} that averages the timeseries of all regions by removing spontaneous oscillations common to the whole brain. Next, we discard the first $40$ seconds of all timeseries to eliminate the effect of initial transients. Finally, we compute the Pearson correlation of the filtered and regressed signals to obtain the synthetic functional connectivity. 
A pipeline describing the above process is illustrated in Fig.~\ref{fig: pipeline}.

\subsection{Application of the clustering control mechanism}

In the remainder of this section, we apply the control method proposed in Section \ref{sec: section 3}. 
We first solve the minimization problem \eqref{eq: control invariance} to find the optimal correction matrix $\Delta^*$ to be applied to $A_\text{original,SC}$ such that condition (C$2$) for the invariance of $\mc S_{\mc P}$ is satisfied. We choose to constrain the corrective action on a set of edges $\mc E_{\mc H} = \mc E\cup \tilde {\mc E}$ that includes the original set $\mc E$ and a minimal set $\tilde {\mc E}$ of randomly selected edges such that problem \eqref{eq: control invariance} is feasible (see Remark \ref{remark: additional constraints}). Fig.~\ref{fig: Delta} and \ref{fig: A+Delta} illustrate the corrective action $\Delta^*$ and the network matrix $A_\text{corrected,SC} = A_\text{original,SC}+\Delta^*$, respectively.

We proceed with the frequency tuning technique for invariance and
stability of $\mc S_{\mc P}$ to $A_{\text{corrected,SC}}$ so that
conditions (C$1$) and (C$3$) are satisfied. The first step involves
computing the mean natural frequency $\omega_\text{av}^{(k)}$ among
all oscillators belonging to the same cluster $\mc P_k$:
$\omega_\text{av}^{(1)} = 199.2$, $\omega_\text{av}^{(2)} =182.9$ and
$\omega_\text{av}^{(3)} = 115.4$ $[$rad/s$]$. Next, we apply the
procedure proposed in Section \ref{sec: frequency tuning}. We plot in
Fig.~\ref{fig: quotient_graph} the spanning tree of the quotient graph
$\mc T_{\mc Q}$, and in Fig.~\ref{fig: alpha_star} the optimal
$\alpha^*$ computed in step (iv) of our procedure. The final natural
frequencies are $\omega^{(1)} = 131.8$, $\omega^{(2)} = 126.4$ and
$\omega^{(3)} = 115.4$ $[$rad/s$]$. Notice that, although our
frequency tuning procedure is sub-optimal, the outcome values remain
well within the range of brain activity frequency bands and, based on
numerical results, outperform the results of Matlab's
\texttt{fmincon} function.

\begin{figure}[t]
  \centering\vspace{-0.2cm}
   \subfigure[]
  {
    \includegraphics[width=0.4\columnwidth]{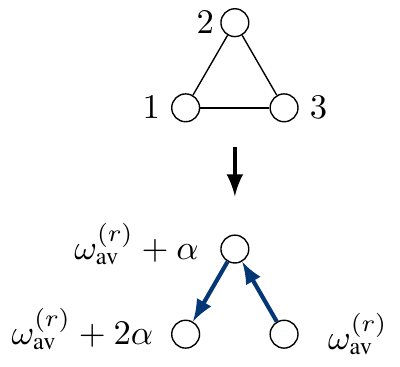}
   \label{fig: quotient_graph}
    }\;\;\;
    \subfigure[]
  {
    \includegraphics[width=0.43\columnwidth]{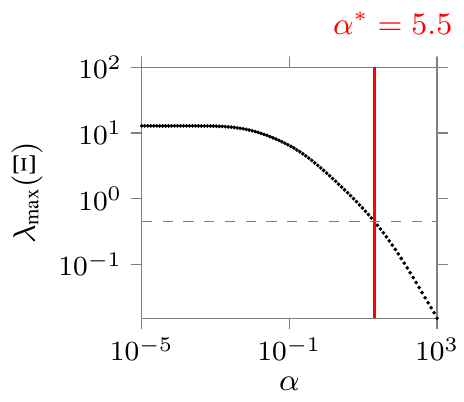}
   \label{fig: alpha_star}
    }
  \caption{Fig.~\ref{fig: quotient_graph} depicts the quotient graph associated with the three clusters in partition $\mc P$ and the natural frequencies that follow from the procedure in Section \ref{sec: frequency tuning}. Fig.~\ref{fig: alpha_star} shows the profile of $\lambda_\text{max}(\Xi)$ as a function of the tuning parameter $\alpha$ on a logarithmic scale. The thick red line highlights the smallest value $\alpha^*$ for which the local stability of the cluster synchronization manifold $\mc S_{\mc P}$ is guaranteed according to condition (C$3$).}\vspace{-0.2cm}
  
\end{figure}

\begin{figure*}[t]
  \centering
  \subfigure[]
  {
    \includegraphics[width=0.475\columnwidth]{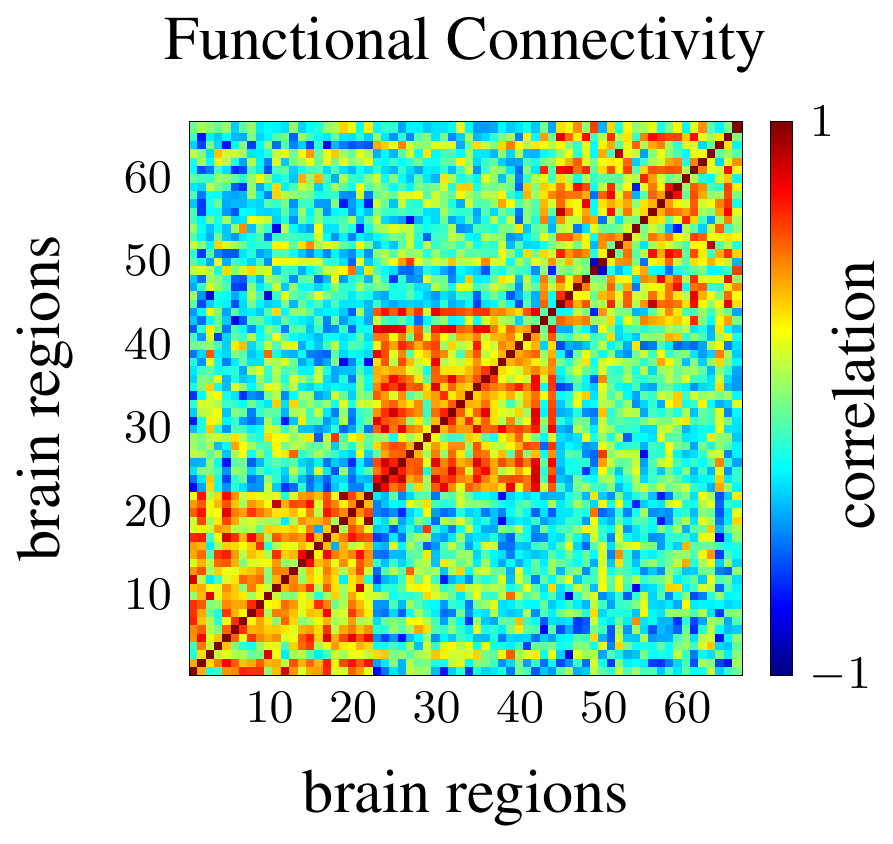}
    \label{fig: CORR}
    }\,
      \subfigure[]
      {
    \includegraphics[width=0.45\columnwidth]{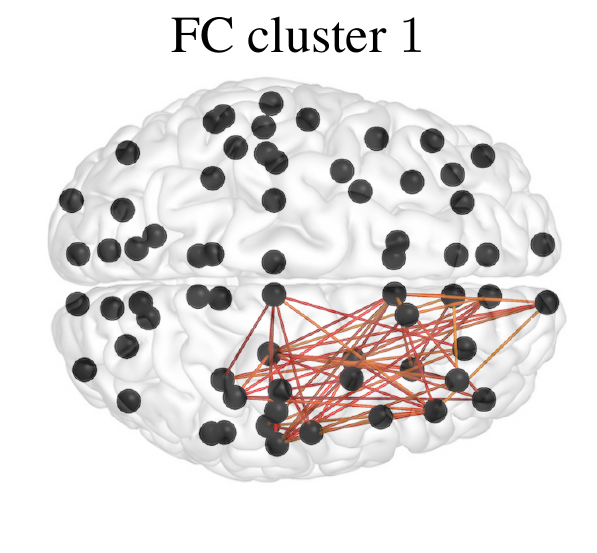}
    \label{fig: C1 BrainNet}
    }\,
      \subfigure[]
      {
    \includegraphics[width=0.45\columnwidth]{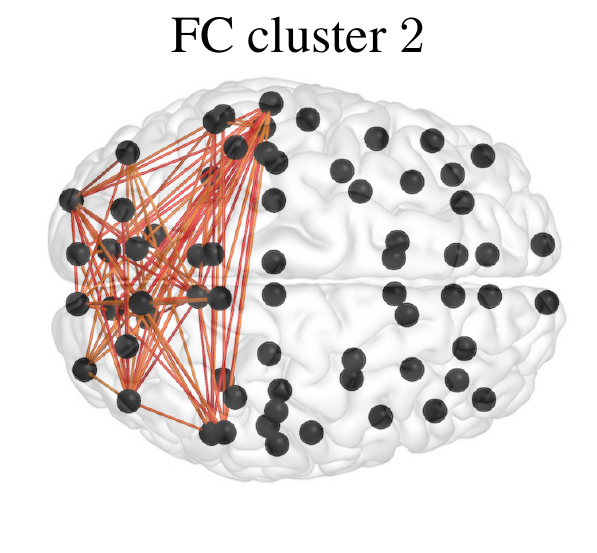}
    \label{fig: C2 BrainNet}
    }\,
      \subfigure[]
      {
    \includegraphics[width=0.45\columnwidth]{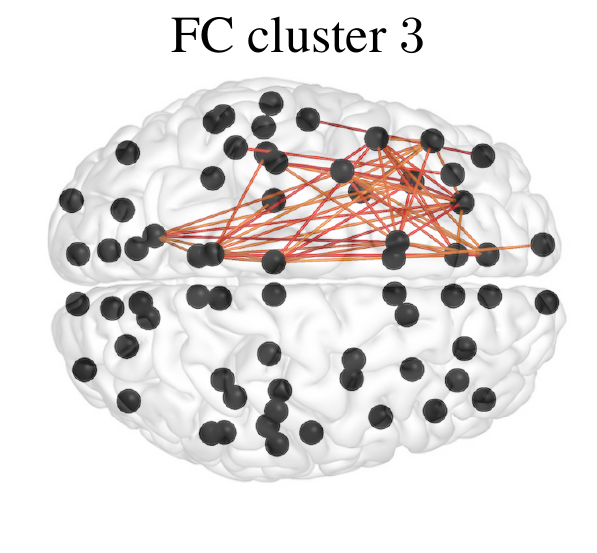}
    \label{fig: C3 BrainNet}
    }
      \vspace{-0.15cm}
  \caption{Fig.~\ref{fig: CORR} represents the correlation matrix that encodes the output functional connectivity (FC) obtained with our control mechanism. Notably, the three clusters are mostly functionally disconnected. That is, there are very few functional connections between nodes belonging to different clusters. This implies that the outcome of our procedure is robust to noisy neural activity and faithfully reproduces synchronized BOLD signals. Fig.~\ref{fig: C1 BrainNet},\ref{fig: C2 BrainNet} and \ref{fig: C3 BrainNet} illustrate the isolated functional connectivity of the desired clusters $\mc P_1$, $\mc P_2$ and $\mc P_3$, respectively, after the correlation matrix has been thresholded to $0.5$ to show only the meaningful functional interactions. The functional edges are color-coded according to the colorbar of Fig.~\ref{fig: CORR}.\vspace{-0.25cm}}
  \label{fig: FC BrainNet}
\end{figure*}

Finally, by following the pipeline described in the previous subsection, we compute the desired functional connectivity pattern, which we show in Fig.~\ref{fig: FC BrainNet}. Notably, the functional connectivity of the desired clusters is strong and the correlations between different clusters are negligible. Thus, the proposed method to control synchronization patterns of oscillatory neural activity lends itself to a physiologically plausible framework and shows rather promising results.

\section{Conclusion}\label{sec: conclusion}
In this work, we propose a minimally invasive technique to obtain
robust synchronization patterns in sparse networks of heterogeneous
Kuramoto oscillators. To the best of our knowledge, this is the first
attempt at blending mathematically rigorous methods with physiological
models of brain activity with the goal of steering whole-brain
synchronization dynamics. Specifically, we cast a constrained
optimization problem whose solution not only satisfies mathematical
conditions for invariance and stability of an arbitrary cluster
synchronization manifold, but also meets biological constraints. We
decompose the complete optimization problem into two simpler
subproblems, and provide efficient methods to solve them. When
applying our technique to correct the network parameters of
empirically-reconstructed anatomical brain data, we find that our
solution, although suboptimal, provides a result that is well within
the range of physiologically plausible parameters. Additionally, we
show that cluster synchronization is robust to small parameter
mismatches and numerical inaccuracies. This result complements
previous prescriptive studies on cluster synchronization and enables
the use of our framework in practical situations.

\appendix

\setcounter{appxlem}{0}
\setcounter{appxthm}{0}
\renewcommand{\theappxlem}{\Alph{section}.\arabic{appxlem}}
\renewcommand{\theappxthm}{\Alph{section}.\arabic{appxthm}}
\renewcommand{\thesubsection}{\Alph{subsection}}

\subsection{Instrumental result for the proof of Theorem \ref{thm correction}}

\begin{appxlem}\label{lemma optimal Delta symm}  
  Consider a network $\mc G = (\mc V, \mc E)$, and an arbitrary (nontrivial) partition $\mc P = \{\mc P_1, \dots, \mc P_m\}$ of $\mc V$. 
  Let $W\in\real^{n\times n}$
  and consider the minimization~problem
  \begin{align}
      \min\limits_{Z} \ &  \|Z-W\|_{\text{F}}^{2} \label{eq: lemma min} \\[.3em]
      \text{s.t.} \ \  &\bar{V}_{\mc P}^{\transpose}Z V_{\mc P}=0,\tag{\theequation a} \label{eq: lemma min a}
  \end{align}
  The minimizer of the problem \eqref{eq: lemma min} subject to \eqref{eq: lemma min a} is
 \begin{align}\label{eq M opt}
  Z^{*}= W - \bar V_{\mc P} \bar V_{\mc P}^{\transpose} W V_{\mc P} V_{\mc P}^{\transpose}.
  \end{align}
  \end{appxlem}\medskip
  \begin{proof}
  We prove the result via the method of Lagrange multipliers. The Lagrangian of \eqref{eq: lemma min} subject to \eqref{eq: lemma min a} is 
$
  \mc L(Z,\Lambda)  = \  \|Z-W\|_{\text{F}}^{2} + \1^{\transpose}(\Lambda \circ \bar{V}_{\mc P}^{\transpose}Z V_{\mc P})\1
  =  \ \mathrm{tr}((Z-W)^{\transpose}(Z-W)) + \mathrm{tr}(\Lambda^{\transpose}\bar{V}_{\mc P}^{\transpose}Z V_{\mc P}),
$
where $\Lambda\in\real^{(n-m)\times m}$ is a matrix of Lagrange multipliers associated with Constraint \eqref{eq: lemma min a}, and in the last equation we used that $\1^{\transpose}(A\circ B)\1 = \mathrm{tr}(A^{\transpose}B)$. Equating the partial derivatives of $\mc L$ to zero yelds:
\begin{align}
 \frac{ \partial \mc L}{\partial Z} &= 2 (Z-W) + \bar V_{\mc P} \Lambda V_{\mc P}^{\transpose}=0, \label{eq opt 1}\\
 \frac{ \partial \mc L}{\partial \Lambda} &= \bar{V}_{\mc P}^{\transpose}Z V_{\mc P}=0, \label{eq opt 2}
\end{align}
We next pre- and post-multiply both sides of \eqref{eq opt 1} by $\bar V_{\mc P}^{\transpose}$ and $V_{\mc P}$, respectively, and obtain
\begin{align}
&  2\bar V_{\mc P}^{\transpose} Z V_{\mc P} = 2 \bar V_{\mc P}^{\transpose} W V_{\mc P} -\bar V_{\mc P}^{\transpose}\bar V_{\mc P} \Lambda V_{\mc P}^{\transpose}V_{\mc P} \notag \\
&\hspace{-0.2cm} \Rightarrow \  2 \bar V_{\mc P}^{\transpose} Z V_{\mc P} = 2\bar V_{\mc P}^{\transpose} W V_{\mc P} -\Lambda \   \Rightarrow \  \Lambda = 2\bar V_{\mc P}^{\transpose} W V_{\mc P}, \label{eq Lambda opt}
\end{align}
where in the second implication we used $V_{\mc P}^{\transpose} V_{\mc P}=I_{n}$, $\bar V_{\mc P}^{\transpose}\bar V_{\mc P}=I_{n-m}$, and $\bar V_{\mc P}^{\transpose}V_{\mc P} =0$, and in the last one we used \eqref{eq opt 2}. Finally, \eqref{eq M opt} follows by substituting \eqref{eq Lambda opt} into~\eqref{eq opt 1}.
  \end{proof}

\subsection{$M$-matrix condition for local stability of $\mc S_{\mc
    P}$}\label{sec: AB}
We now recall a stability condition established in
\cite{TM-GB-DSB-FP:19a}.  We let $x_{ij} = \theta_j-\theta_i$ denote
the phase difference between oscillators $i$ and $j$, and
$x_\text{intra} =
[x_\text{intra}^{(1)},\dots,x_\text{intra}^{(m)}]^\transpose$ denote a
smallest set of intra-cluster differences akin to (see
\cite{TM-GB-DSB-FP:19a}), where $x_\text{intra}^{(k)}$ contains only
phase differences of oscillators in $\mc P_k$.\footnote{The definition
  in \cite{TM-GB-DSB-FP:19a} is given for undirected graphs. However,
  it is straightforward to extend all definitions to digraphs,
  provided that the quotient graph $\mc Q$ defined in Section
  \ref{sec: frequency tuning} is strongly connected. This is not a
  restrictive assumption in the context of structural brain networks.}
Further, let $J_k$ be the Hurwitz stable Jacobian matrix of the
intra-cluster phase difference dynamics of $x_\text{intra}^{(k)}$.
\begin{appxthm}{\bf \emph{(Sufficient condition on network weights for
      the stability of $\mc S_{\mc P}$ \cite{TM-GB-DSB-FP:19a})}}\label{thm: stability of Sp}
  Let $\kappa = 2\max_{r}|\mc P_r|-2 $, and
  \begin{align*}
    \gamma^{(k\ell)} = \begin{cases} 
      \displaystyle \kappa \sum_{j\in \mc P_{\ell}}
      a_{ij}, &  \text{ if } \ell\ne k, \\ 
      \displaystyle \kappa  \sum_{\substack{
          \ell\ne k}}\sum_{j\in \mc P_{\ell}} a_{ij}, & \text{ otherwise,} 
    \end{cases}
  \end{align*}
  with $k, \ell \in \until{m}$, $i \in \mc P_k$. Define the $m\times m$ matrix~$S$:
    \vspace{-0.1cm}
  \begin{align*}
    S = [s_{k\ell}] = \begin{cases}
      \lambda_{\text{max}}^{-1}(X_{k})-\gamma^{(kk)}& \text{ if } k=\ell,\\
      -\gamma^{(k\ell)} & \text{ if } k\neq \ell,
    \end{cases}
  \end{align*}
  where $X_{k}\succ0$ satisfies
  $J_{k}^{\transpose}X_{k}+X_{k}J_{k}=-I$. If $S$ is an $M$-matrix, then the cluster
  synchronization manifold is locally exponentially stable.
\end{appxthm}
\smallskip

Theorem \ref{thm: stability of Sp} shows that the cluster
synchronization manifold is stable when the intra-cluster dynamics are
sufficiently more attractive than the inter-cluster couplings.
  
\renewcommand{\baselinestretch}{.985}

\bibliographystyle{unsrt}
\bibliography{alias,FP,Main,New}

\end{document}